\documentclass[singlecolumn]{revtex4}  
\usepackage[lofdepth,lotdepth]{subfig}
\usepackage{graphicx}
\usepackage{dcolumn}
\usepackage{bm}
\usepackage{multirow}
\usepackage{pgfplots}
\usepackage{amssymb,amsmath,amsthm}
\usepackage{mathtools}
\usepackage{anysize}
\usepackage{hyperref}
\pgfplotsset{compat=newest}

\marginsize{0.8in}{0.8in}{0.8in}{0.8in}

\pgfdeclarelayer{background}
\pgfdeclarelayer{foreground}
\pgfsetlayers{background,main,foreground}

\theoremstyle{plain}
\newtheorem{theorem}{Theorem}
\newtheorem{corollary}{Corollary}
\newtheorem{lemma}{Lemma}

\renewenvironment{proof}{\noindent{\bf Proof:}\;}{$\square$\,}

\global\long\def\H{\textsc{H}}
\global\long\def\P{\textsc{P}}
\global\long\def\Z{\textsc{Z}}
\global\long\def\CNOT{\textsc{CNOT}}

\global\long\def\CZ{\textsc{CZ}}

\usepackage{xy}
\xyoption{matrix}
\xyoption{frame}
\xyoption{arrow}
\xyoption{arc}

\usepackage{ifpdf}
\ifpdf
\else
\PackageWarningNoLine{Qcircuit}{Qcircuit is loading in Postscript mode.  The Xy-pic options ps and dvips will be loaded.  If you wish to use other Postscript drivers for Xy-pic, you must modify the code in Qcircuit.tex}
\xyoption{ps}
\xyoption{dvips}
\fi

\entrymodifiers={!C\entrybox}

\newcommand{\ket}[1]{{\left\vert{#1}\right\rangle}}
\newcommand{\qw}[1][-1]{\ar @{-} [0,#1]}
\newcommand{\qwx}[1][-1]{\ar @{-} [#1,0]}

\newcommand{\gate}[1]{*+<.6em>{#1} \POS ="i","i"+UR;"i"+UL **\dir{-};"i"+DL **\dir{-};"i"+DR **\dir{-};"i"+UR **\dir{-},"i" \qw}

\newcommand{\control}{*!<0em,.025em>-=-<.2em>{\bullet}}

\newcommand{\ctrl}[1]{\control \qwx[#1] \qw}

\newcommand{\targ}{*+<.02em,.02em>{\xy ="i","i"-<.39em,0em>;"i"+<.39em,0em> **\dir{-}, "i"-<0em,.39em>;"i"+<0em,.39em> **\dir{-},"i"*\xycircle<.4em>{} \endxy} \qw}
\newcommand{\lstick}[1]{*!R!<.5em,0em>=<0em>{#1}}

\newcommand{\Qcircuit}{\xymatrix @*=<0em>}

\begin{document}

\title{Optimized Aaronson-Gottesman stabilizer circuit simulation \\ through quantum circuit transformations\footnote{Subsumed by \href{https://arxiv.org/abs/1705.09176}{arXiv:1705.09176}.}}
\author{Dmitri Maslov}
\email{dmitri.maslov@gmail.com}
\affiliation{National Science Foundation, Arlington, VA 22230, USA}

\begin{abstract}
In this paper we improve the layered implementation of arbitrary stabilizer circuits introduced by Aaronson and Gottesman in {\it Phys. Rev. A 70(052328)}, 2004.  In particular, we reduce their 11-stage computation -H-C-P-C-P-C-H-P-C-P-C- into an 8-stage computation of the form -H-C-CZ-P-H-P-CZ-C-.  We show arguments in support of using -CZ- stages over the -C- stages: not only the use of -CZ- stages allows a shorter layered expression, but -CZ- stages are simpler and appear to be easier to implement compared to the -C- stages.  Relying on the 8-stage decomposition we develop a two-qubit depth-$(14n-4)$ implementation of stabilizer circuits over the gate library $\{\P,\H,\CNOT\}$, executable in the LNN architecture, improving best previously known depth-$25n$ circuit, also executable in the LNN architecture.  Our constructions rely on folding arbitrarily long sequences $($-P-C-$)^m$ into a 3-stage computation -P-CZ-C-, as well as efficient implementation of the -CZ- stage circuits.
\end{abstract}

\maketitle

\section{Introduction}
Stabilizer circuits are of particular interest in quantum information processing (QIP) due to their prominent role in fault tolerance \cite{ar:crss, www:g, bk:nc, ar:s}, the study of entanglement \cite{ar:bdsw, bk:nc}, and in evaluating quantum information processing proposals via randomized benchmarking \cite{ar:klr}, to name a few.  For the purpose of this paper, we define stabilizer circuits to be those composed with the single-qubit quantum Hadamard, \H$:=\frac{1}{\sqrt{2}}\left(\begin{array}{rr}
1 & 1\\
1 & -1
\end{array}\right)$, and Phase gates, $\P:=\begin{pmatrix} 1 & 0 \\ 0 & i \end{pmatrix}$, and the entangling \CNOT, \CNOT$:=\begin{pmatrix} 1 & 0 & 0 & 0 \\ 0 & 1 & 0 & 0 \\ 0 & 0 & 0 & 1 \\ 0 & 0 & 1 & 0 \end{pmatrix}$, and \CZ, \CZ$:=\begin{pmatrix*}[r] 1 & 0 & 0 & 0 \\ 0 & 1 & 0 & 0 \\ 0 & 0 & 1 & 0 \\ 0 & 0 & 0 & -1 \end{pmatrix*}$ gates. 
We reuse the notations from \cite{ar:ag}, and assume reader's familiarity with the concepts and results of this paper, as it is essential for understanding our result.  

The stabilizer circuits over $n$ qubits, such as defined above form a finite group.  It is equivalent to the proper size symplectic group, which is best seen through considering tableau representation developed in \cite{ar:ag}.  A square block matrix $M=\begin{pmatrix} A & B \\ C & D \end{pmatrix}$ is called symplectic iff the following three conditions hold: $A^TC=C^TA$, $B^TD=D^TB$, and $A^TD-C^TB=I$.  The $n$-qubit stabilizer group is equivalent to the group of binary $2n \times 2n$ symplectic matrices, $Sp(2n,F_2)$.  This reduction allows to effortlessly evaluate the stabilizer group size, through employing the well-known formula to calculate the number of elements in the respective symplectic group, $|Sp(2n,F_2)|={2^{n^2}\prod\limits_{j=1}^{n}(2^{2j}-1)}=2^{2n^2+O(n)}$. 

In this paper, we rely on the phase polynomial representation of $\{\P,\CNOT\}$ circuits.  Specifically, arbitrary quantum circuits over \textsc{P} and \textsc{CNOT} gates can be described in an alternate form, which we refer to as {\em phase polynomial description}, and vice versa, each phase polynomial description can be written as a \textsc{P} and \textsc{CNOT} gate circuit.  We use this result to induce circuit transformations via rewriting the respective phase polynomials.  We adopt the phase polynomial expression result from \cite{ar:ammr} to this paper as follows:

\begin{theorem}\label{thm:0}
Any circuit $C$ on $n$ qubits over $\{\P,\CNOT\}$ library with $k$ Phase gates can be described by the combination of a phase polynomial $p=f_1(x_1, x_2, ..., x_n) + f_2(x_1, x_2, ..., x_n) + \cdots + f_k(x_1, x_2, ..., x_n)$ and a linear reversible function $g(x_1, x_2, ..., x_n)$, such that the action of $C$ can be constructed as $C|x_1x_2...x_n\rangle = i^p|g(x_1, x_2, ..., x_n)\rangle$, where $i$ is the complex number $i$. Functions $f_j$ corresponding to the $j^{\text{th}}$ Phase gate are obtained from the circuit $C$ via devising Boolean linear functions computed by the \textsc{CNOT} gates in the circuit $C$ leading to the position of the respective Phase gate. 
\end{theorem}

In the following we focus on finding a short layered sequence of gates capable of representing an arbitrary stabilizer circuit over $n$ primary inputs.  The layers are defined as follows: 
\begin{itemize}
\item -H- layer contains all unitaries representable by arbitrary circuits composed of the Hadamard gates. Since $\H^2=Id$, and Hadamard gate is a single-qubit gate, -H- layer has zero or one gates acting on each of the respective qubits.  The number of distinct layers -H- on $n$ qubits is thus $2^n$.  We say -H- has $n$ Boolean degrees of freedom.
\item -P- layer is composed of an arbitrary set of Phase gates.  Since $\P^4=Id$, and Phase gate is also a single-qubit gate, -P- layer has anywhere between zero to three gates on each of the respective qubits.  Note that $\P^2=\Z$, and therefore the gate sequence $\P\P$ is may be better implemented as the Pauli-$\Z$ gate; $\P^3=\P^\dagger$, and frequently $\P^\dagger$ is constructible with the same cost as $\P$.  This means that -P- layer is essentially analogous to the -H- layer in the sense that it consists of at most $n$ individual single-qubit gates.  The number of different unitaries represented by -P- layers on $n$ qubits is $2^{2n}$.  We say -P- has $2n$ Boolean degrees of freedom.  
\item -C- layer contains all unitaries computable by the $\CNOT$ gates.  The number of different -C- layers corresponds to the number of affine linear reversible functions, and it is equal to $\prod\limits_{j=0}^{n-1}(2^n-2^j)=2^{n^2+O(n)}$ \cite{ar:pmh}.  We say -C- has $n^2+O(n)$ Boolean degrees of freedom.
\item -CZ- layer contains all unitaries computable by the $\CZ$ gates.  Since all $\CZ$ gates commute, and due to $\CZ$ being self-inverse, {\em i.e.}, $\CZ^2=Id$, the number of different unitaries computable by -CZ- layers is $\prod\limits_{j=1}^{n}2^{n-j}=2^{\frac{n^2}{2}+O(n)}$.  We say -CZ- has $\frac{n^2}{2}+O(n)$ Boolean degrees of freedom.
\end{itemize} 

Observe that the above count of the degrees of freedom suggests that -P- and -H- layers are ``simple''.  Indeed, each requires no more than the linear number of single-qubit gates to be constructed via a circuit.  The number of the degrees of freedom in -C- and -CZ- stages is quadratic in $n$.  Other than two-qubit gates often being more expensive than the single-qubit gates \cite{ar:deb,www:IBM}, the comparison of the degrees of freedom suggests that we will need more of the respective gates to construct each such stage.  The -CZ- layer has roughly half the number of the degrees of freedom compared to the -C- layer.  We may thus reasonably expect that the -CZ- layer can be easier to obtain.  

Unlike the -C- circuits, the -CZ- circuits have not been studied in the literature.  Part of the reason could be due to $\CZ$ gate complexity of the -CZ- circuits being a very inconspicuous to study problem: indeed, worst case optimal circuit has $\frac{(n-1)n}{2}$ $\CZ$ gates, and optimal circuits are easy to construct, as they are determined by the presence of lack of $\CZ$ gates acting on the individual pairs of qubits.  However, we claim that using only $\CZ$ gates to construct -CZ- layer is not the best solution, and a better approach would be to also employ the $\CNOT$ and $\P$ gates as well.  Indeed, both $\CNOT$ and $\CZ$ gates must have a comparable cost of the implementation, since they are related by the formula $\CNOT(a,b)=\H(b)\CZ(a,b)\H(b)$, and single-qubit gates are ``easy'' \cite{www:IBM, ar:deb}.  $\CZ$ is furthermore the elementary gate in superconducting QIP \cite{ar:ggz}, and as such, technically, costs less than the $\CNOT$, and in the trapped ions QIP (Quantum Information Processing) the costs of the two are comparable \cite{ar:m}.  Further discussion of the relation of implementation costs between -C- and -CZ- layers is postponed to Section \ref{sec:cvscz}.

The different layers can be interleaved to obtain stabilizer circuits not computable by a single layer.  A remarkable result of \cite{ar:ag} shows that 11 stages over a computation of the form -H-C-P-C-P-C-H-P-C-P-C- suffices to compute an arbitrary stabilizer circuit.  The number of Boolean degrees of freedom in the group of stabilizer unitaries, defined as the logarithm base-2 of their total count, is given by the formula $\log_2{|Sp(2n,F_2)|}=2n^2+O(n)$.  This suggests that the 11-stage circuit by Aaronson and Gottesman \cite{ar:ag} could be suboptimal, as it relies on $5n^2+O(n)$ degrees of freedom, whereas only $2n^2+O(n)$ are necessary.  Indeed, we find a shorter 8-stage decomposition of the form -H-C-CZ-P-H-P-CZ-C- relying on $3n^2+O(n)$ degrees of freedom. 

\section{$($-P-C-$)^m$ circuits}
In this section we show that an arbitrary length $n$-qubit computation described by the stages -P-C-P-C-...-P-C- folds into an equivalent three-stage computation -P-CZ-C-.

\begin{theorem}\label{thm:1}
$(-P-C-)^m=-P-CZ-C-$.
\end{theorem}
\begin{proof}
$($-P-C-$)^m$ circuit has $k\leq nm$ Phase gates.  Name those gates $\P_{j=1..k}$, denote Boolean linear functions they apply phases to as $f_{j=1..k}(x_1,x_2,...,x_n)$, and name the reversible linear function computed by \linebreak $($-P-C-$)^m$ (Theorem \ref{thm:0}) as $g(x_1,x_2,...,x_n)$. Phase polynomial computed by the original circuit is $f_1(x_1,x_2,...,x_n)+f_2(x_1,x_2,...,x_n)+...+f_k(x_1,x_2,...,x_n)$.  We will next transform phase polynomial to an equivalent one, that will be easier to write as a compact circuit.  To accomplish this, observe that $i^{a+b+c+(a\oplus b) + (a\oplus c) + (b\oplus c) + (a\oplus b \oplus c)} = i^4= 1$, where $a$, $b$, and $c$ are arbitrary Boolean linear functions of the primary variables.  This equality can be verified by inspection through trying all 8 possible combinations for Boolean values $a$, $b$, and $c$. The equality can be rewritten as 
\begin{equation}\label{phasetmpl}
i^{a\oplus b \oplus c} = i^{3a+3b+3c+3(a\oplus b)+3(a\oplus c)+3(b\oplus c)},
\end{equation} 
suggesting how it will be used.  The following algorithm takes $n-2$ steps.
\begin{enumerate}
\item[Step $n$.]  Consolidate terms in the phase polynomial $f_1(x_1,x_2,...,x_n)+f_2(x_1,x_2,...,x_n)+...+f_k(x_1,x_2,...,x_n)$ by replacing $uf_j(x_1,x_2,...,x_n)+vf_k(x_1,x_2,...,x_n)$ with $(u+v \bmod 4)f_j(x_1,x_2,...,x_n)$ whenever $f_j=f_k$.  Once done, look for $f_j=x_1 \oplus x_2 \oplus ... \oplus x_n$, being the maximal length linear function of the primary inputs.  If no such function found, move to the next step.  If it is found with a non-zero coefficient $u$, as an additive term $u(x_1 \oplus x_2 \oplus ... \oplus x_n)$, replace it by the equivalent 6-term mixed arithmetic polynomial $(4-u)x_1 + (4-u)x_2 + (4-u)(x_3 \oplus x_4 \oplus ... \oplus x_n) + (4-u)(x_1 \oplus x_2) + (4-u)(x_1 \oplus x_3 \oplus x_4 \oplus ... \oplus x_n) + (4-u)(x_2 \oplus x_3 \oplus ... \oplus x_n)$. This transformation is derived from (\ref{phasetmpl}) by assigning $a=x_1$, $b=x_2$, and $c=x_3 \oplus ... \oplus x_n$.  Consolidate all equal terms.  The transformed phase polynomial is equivalent to the original one in the sense of the overall combination of phases it prescribes to compute, however, it is expressed over linear terms with at most $n-1$ variables.
\item[Step $s$,] $s=(n-1)..3$. From the previous step we have phase polynomial of the form $u^\prime_1f^\prime_1(x_1,x_2,...,x_n)+u^\prime_2f^\prime_2(x_1,x_2,...,$ $x_n)+...+u^\prime_{k^\prime}f^\prime_{k^\prime}(x_1,x_2,...,x_n).$ By construction it is guaranteed that the functions $f^\prime_{j=1..k^\prime}$ EXOR no more than $s$ literals.  For each $f^\prime_j=x_{j_1} \oplus x_{j_2} \oplus ... \oplus x_{j_s}$, with the coefficient $u^{\prime}_j \not\equiv 0 \bmod 4$ replace this term with the sum of six terms, each having no more than $s-1$ literals by using the equation (\ref{phasetmpl}) and setting $a$, $b$, and $c$ to carry linear functions over the non-overlapping non-empty subsets of $\{x_{j_1}, x_{j_2}, ..., x_{j_s}\}$ whose union gives the entire set $\{x_{j_1}, x_{j_2}, ..., x_{j_s}\}$. Value $s=3$ marks the last opportunity to break down a term in the phase polynomial expression into a set of terms over smaller numbers of variables.  Upon completion of this step, the linear functions participating in the phase polynomial expression contain at most two literals each.
\end{enumerate}
The transformed phase polynomial description of the original circuit now has the following form: 
phase polynomial $\sum\limits_{j=1}^{n}u_jx_j + \sum\limits_{j=1}^{n}\sum\limits_{k=j+1}^{n}u_{j,k}(x_j \oplus x_k)$, where $u_{\cdot}, u_{\cdot,\cdot} \in \mathbb{Z}_4$, and the linear reversible function $g(x_1,x_2,...,x_n)$. We next show how to implement such a unitary as a -P-CZ-C- circuit, focusing separately on the phase polynomial and the linear reversible part. We synthesize individual terms in the phase polynomial as follows. 
\begin{itemize}
\item For $j=1..n,$ $u_jx_j$ is obtained as the single-qubit gate circuit $\P^{u_j}(x_j)$;
\item For $j=1..n$, $k=j+1..n$, $u_{j,k}(x_j \oplus x_k)$ is obtained as follows: 
\begin{itemize}
\item if $u_{j,k}\equiv 2 \bmod 4$, by the circuit $\P^2(x_j)\P^2(x_k) = \Z(x_j)\Z(x_k)$;
\item if $u_{j,k}\equiv 1 \text{ or } 3 \bmod 4$, by the circuit  $\P^{u_{j,k}}(x_j)\P^{u_{j,k}}(x_k) \CZ(x_j,x_k)$.
\end{itemize} 
\end{itemize}
The resulting circuit contains $\P$ and $\CZ$ gates; it implements phase polynomial $\sum\limits_{j=1}^{n}u_jx_j + \sum\limits_{j=1}^{n}\sum\limits_{k=j+1}^{n}u_{j,k}(x_j \oplus x_k)$ and the identity linear reversible function. Since all $\P$ and $\CZ$ gates commute, Phase gates can be collected on the left side of the circuit. This results in the ability to express phase polynomial construction as a -P-CZ- circuit.  We conclude the entire construction via obtaining the linear reversible function $g(x_1,x_2,...,x_n)$ as a -C- stage, with the overall computation described as a -P-CZ-C- circuit.
\end{proof}

Note that -P-CZ-C- can also be written as -C-P-CZ-, if one first synthesizes the linear reversible function $g(x_1,x_2,...,x_n)=(g_1(x_1,x_2,...,x_n),g_2(x_1,x_2,...,x_n),...,g_n(x_1,x_2,...,x_n))$, and expresses the phase polynomial in terms of degree-2 terms over the set $\{g_1,g_2,...,g_n\}$. Other ways to write such a computation include -CZ-P-C- and -C-CZ-P-, that are obtained from the first two by commuting -P- and -CZ- stages.

\begin{corollary}
-H-C-P-C-P-C-H-P-C-P-C- \cite{ar:ag} = -H-C-CZ-P-H-P-CZ-C-.
\end{corollary}

\section{-C- vs -CZ-}\label{sec:cvscz}

We have previously noted that $\CNOT$ and $\CZ$ gates have a comparable cost as far as their implementation within some QIP proposals is concerned.  In this section, we study  $\{\P, \CZ, \CNOT\}$ implementations of the stages -C- and -CZ-.  The goal is to provide further evidence in support of the statement that -CZ- can be thought of as a simpler stage compared to the -C- stage, and going beyond counting the degrees of freedom argument.

\begin{lemma}
Optimal quantum circuit over $\{\CZ\}$ library for a -CZ- stage has at most $\frac{n(n-1)}{2}$ $\CZ$ gates.
\end{lemma}

Indeed, all $\CZ$ gates commute, which limits the expressive power of the circuits over $\CZ$ gates.  However, once we add the non-commuting $\CNOT$ gate, and after that the Phase gate, the situation changes.  We can now implement -CZ- circuits more efficiently, such as illustrated by the following circuit identities. 
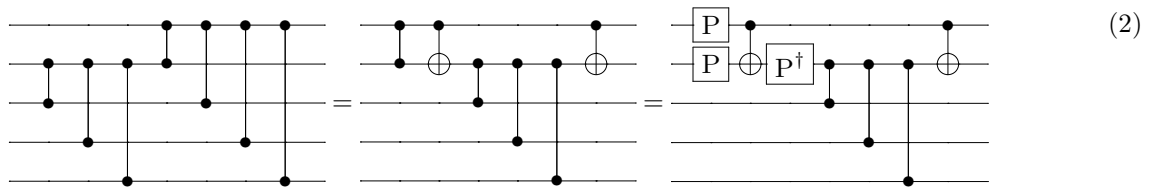
\begin{equation}\label{circ:1}
\Qcircuit @C=1.23em @R=1.23em @! {
& \qw 		& \qw 		& \qw 		& \ctrl{1}	& \ctrl{2}  & \ctrl{3}	& \ctrl{4}	& \qw \\
& \ctrl{1}  & \ctrl{2}	& \ctrl{3}	& \ctrl{-1}	& \qw 		& \qw 		& \qw 		& \qw \\
& \ctrl{-1} & \qw 		& \qw 		& \qw 		& \ctrl{-2} & \qw 		& \qw 		& \qw \\
& \qw  		& \ctrl{-2}	& \qw		& \qw 		& \qw  		& \ctrl{-3}	& \qw		& \qw \\
& \qw 		& \qw 		& \ctrl{-3}	& \qw 		& \qw 		& \qw 		& \ctrl{-4}	& \qw 
} 
\raisebox{-2.95em}{\hspace{1mm}=\hspace{1mm}}
\Qcircuit @C=0.65em @R=0.65em @! {
& \ctrl{1}	& \ctrl{1} 	& \qw 		& \qw 		& \qw 		& \ctrl{1} 	& \qw \\
& \ctrl{-1}	& \targ		& \ctrl{1}  & \ctrl{2}	& \ctrl{3}	& \targ		& \qw \\
& \qw 		& \qw		& \ctrl{-1} & \qw 		& \qw 		& \qw		& \qw \\
& \qw 		& \qw		& \qw  		& \ctrl{-2}	& \qw		& \qw		& \qw \\
& \qw 		& \qw		& \qw 		& \qw 		& \ctrl{-3}	& \qw		& \qw 
} 
\raisebox{-2.95em}{\hspace{1mm}=\hspace{1mm}}
\Qcircuit @C=-0.28em @R=-0.28em @! {
& \gate{\P}	& \ctrl{1} 	& \qw 				& \qw 		& \qw 		& \qw 		& \ctrl{1} 	& \qw \\
& \gate{\P}	& \targ		& \gate{\P^\dagger}	& \ctrl{1}  & \ctrl{2}	& \ctrl{3}	& \targ		& \qw \\
& \qw 		& \qw		& \qw 				& \ctrl{-1} & \qw 		& \qw 		& \qw		& \qw \\
& \qw 		& \qw		& \qw 				& \qw  		& \ctrl{-2}	& \qw		& \qw		& \qw \\
& \qw 		& \qw		& \qw 				& \qw 		& \qw 		& \ctrl{-3}	& \qw		& \qw 
} 
\end{equation}
The unitary implemented by the circuitry shown above requires $7$ $\CZ$ gates as a $\{\CZ\}$ circuit, $6$ gates as a $\{\CZ,\CNOT\}$ circuit, and only $5$ two-qubit gates as a $\{\P,\CZ,\CNOT\}$ circuit.  This illustrates that the $\CNOT$ and $\P$ gates are important in constructing efficient -CZ- circuits.

We may consider adding the $\P$ and $\CZ$ gates to the $\{\CNOT\}$ library in hopes of constructing more efficient circuits implementing the -C- stage.  However, this does not help.
\begin{lemma}
Any $\{\P, \CZ, \CNOT\}$ circuit implementing an element of the layer -C- using a non-zero number of $\P$ and $\CZ$ gates is suboptimal.
\end{lemma}
\begin{proof}
Each $\P$ gate applied to a qubit $x$ can be expressed as a phase polynomial $1\cdot x$ over the identity reversible linear function.  Each $\CZ$ gate applied to a set of qubits $y$ and $z$ can be expressed as a phase polynomial $y + z + 3(y \oplus z)$ and the identity reversible function.  Removing all $\P$ and $\CZ$ gates from the given circuit thus modifies only the phase polynomial part of its phase polynomial description.  Removing all $\P$ and $\CZ$ gates from the $\{\P, \CZ, \CNOT\}$ circuit guarantees that the phase polynomial of the resulting circuit equals to the identity, such as required in the -C- stage.  This results in the construction of a shorter circuit in cases when the original $\P$ and $\CZ$ gate count was non-zero.
\end{proof}

We next show a table of optimal counts and upper bounds on the number of gates it takes to synthesize most difficult function from the sets -CZ- and -C- for some small $n$, Table \ref{tab:main}.  Observe how the two-qubit gate counts for -CZ- stage, when constructed as a circuit over  $\{\P, \CZ, \CNOT\}$ library, remain lower than those for the -C- stage. 

\begin{table}[ht] 
\begin{tabular}{c|cc|cc} 
  & \multicolumn{2}{c|}{-CZ-} 		& \multicolumn{2}{c}{-C-} \\
n & $\{\CZ\}$ 	& $\{\P, \CZ,\CNOT\}$ 	& $\{\CNOT\}$ 	& $\{\P, \CZ,\CNOT\}$ \\ \hline
2 & 1			& 1					& 3				& 3 \\
3 & 3			& 3					& 6 			& 6  \\
4 & 6			& 5					& 9 			& 9  \\
5 & 10			& 7					& 12			& 12  \\
\end{tabular} 
\caption{Gate counts required to implement arbitrary -CZ- and -C- stages for some small $n$: optimal -CZ- stage gate counts as circuits over $\{\CZ\}$, upper bounds on the two-qubit gate count for -CZ- over $\{\P, \CZ,\CNOT\}$ (achieved based on the application of identities (\ref{circ:1}) applied to circuits with $\CZ$ gates), optimal $\{\CNOT\}$ and $\{\P, \CZ,\CNOT\}$ two-qubit gate counts for stage -C-.} \label{tab:main}
\end{table}

Literature encounters an asymptotically optimal algorithm \cite{ar:pmh} for $\{\CNOT\}$ synthesis of arbitrary -C- stage functions.  The gate complexity is $O\left(\frac{n^2}{\log n}\right)$.  It is possible that an asymptotically optimal algorithm for $\{\P, \CZ, \CNOT\}$ circuits implementing arbitrary -CZ- stage functions can be developed, at which point its complexity has to be $O\left(\frac{n^2}{\log n}\right)$.  To determine which of the two results in shorter circuits, one has to develop constants in front of the leading complexity terms. 

\vspace{1mm}
Gate count is only one possible metric of efficiency.  Two-qubit gate depth over Linear Nearest Neighbour (LNN) architecture is also a very important metric of efficiency.  This metric has been applied in \cite{ar:kms} to show an asymptotically optimal upper bound of $5n$ $\CNOT$ layers required to obtain an arbitrary -C- stage.  

Define $\widehat{\text{-CZ-}}$ to be -CZ- accompanied by the complete qubit reversal (defined as the linear reversible mapping $\ket{x_1x_2...x_n} \mapsto \ket{x_nx_{n-1}...x_1}$).  We next show that $\widehat{\text{-CZ-}}$ can be executed as a two-qubit depth-$(2n+2)$ computation over LNN.  This result will be used to reduce depth in the implementation of arbitrary stabilizer circuits.

\begin{theorem}\label{thm:2}
$\widehat{\text{-CZ-}}$ can be implemented as a $\CNOT$ depth-$(2n+2)$ circuit.
\end{theorem}
\begin{proof}

\begin{figure}[t]
\centerline{
\Qcircuit @C=0.97em @R=0.97em @! {
&&&&& [6,\;\;] &&&&&& [\;\;, \;\;]	&&&&&& [\;\;,3] &&&&&&& \\
&&&&& [4,\;\;] &&&&&& [\;\;, \;\;]	&&&&&& [\;\;,3] &&&&&&& \\
&&&&& [4,\;\;] &&&&&& [6, 3]		&&&&&& [\;\;,5] &&&&&&& \\
&&&&& [2,\;\;] &&&&&& [4, 3]		&&&&&& [\;\;,5] &&&&&&& \\
\lstick{[1,1]}	& \ctrl{1}	& \qw 		& \targ		& \qw		& [2,3] &
				& \ctrl{1}	& \qw 		& \targ		& \qw		& [4,5] &
				& \ctrl{1}	& \qw 		& \targ		& \qw		& [6,7] & 
				& \ctrl{1}	& \qw 		& \targ		& \qw		& \qw & {[7,7]} \\
\lstick{[2,2]}	& \targ 	& \targ		& \ctrl{-1}	& \ctrl{1}	& [1,3] &
				& \targ 	& \targ		& \ctrl{-1}	& \ctrl{1}	& [2,5] &
				& \targ 	& \targ		& \ctrl{-1}	& \ctrl{1}	& [4,7] &
				& \targ 	& \targ		& \ctrl{-1}	& \ctrl{1}	& \qw & {[6,6]} \\
\lstick{[3,3]}	& \ctrl{1}	& \ctrl{-1}	& \targ		& \targ		& [1,5] &
				& \ctrl{1}	& \ctrl{-1}	& \targ		& \targ		& [2,7] &
				& \ctrl{1}	& \ctrl{-1}	& \targ		& \targ		& [4,6] & 
				& \ctrl{1}	& \ctrl{-1}	& \targ		& \targ		& \qw & {[5,5]} \\
\lstick{[4,4]}	& \targ 	& \targ		& \ctrl{-1}	& \ctrl{1}	& [3,5] &
				& \targ 	& \targ		& \ctrl{-1}	& \ctrl{1}	& [1,7] &
				& \targ 	& \targ		& \ctrl{-1}	& \ctrl{1}	& [2,6] & 
				& \targ 	& \targ		& \ctrl{-1}	& \ctrl{1}	& \qw & {[4,4]} \\
\lstick{[5,5]}	& \ctrl{1}	& \ctrl{-1}	& \targ		& \targ		& [3,7] &
				& \ctrl{1}	& \ctrl{-1}	& \targ		& \targ		& [1,6] &
				& \ctrl{1}	& \ctrl{-1}	& \targ		& \targ		& [2,4] &
				& \ctrl{1}	& \ctrl{-1}	& \targ		& \targ		& \qw & {[3,3]} \\
\lstick{[6,6]}	& \targ 	& \targ		& \ctrl{-1}	& \ctrl{1}	& [5,7] &
				& \targ 	& \targ		& \ctrl{-1}	& \ctrl{1}	& [3,6] &
				& \targ 	& \targ		& \ctrl{-1}	& \ctrl{1}	& [1,4] &
				& \targ 	& \targ		& \ctrl{-1}	& \ctrl{1}	& \qw & {[2,2]} \\
\lstick{[7,7]}	& \qw		& \ctrl{-1}	& \qw		& \targ		& [5,6] &
				& \qw		& \ctrl{-1}	& \qw		& \targ		& [3,4] &
				& \qw		& \ctrl{-1}	& \qw		& \targ		& [1,2] & 
				& \qw		& \ctrl{-1}	& \qw		& \targ		& \qw & {[1,1]} \\
&&&&& [\;\;,6] &&&&&& [5, 4]		&&&&&& [3,\;\;] &&&&&&& \\
&&&&& [\;\;,4] &&&&&& [5, 2]		&&&&&& [3,\;\;] &&&&&&& \\
&&&&& [\;\;,4] &&&&&& [\;\;, \;\;]	&&&&&& [5,\;\;] &&&&&&& \\
&&&&& [\;\;,2] &&&&&& [\;\;, \;\;]	&&&&&& [5,\;\;] &&&&&&& \\
&&&&& \downarrow \;\; \uparrow &&&&&& \downarrow \;\; \uparrow	&&&&&& \downarrow \;\; \uparrow &&&&&&&
}
}
\caption{Constructing $\widehat{\text{-CZ-}}$ for $n=7$. The circuit uses $m+1=\frac{n-1}{2}+1=4$ depth-$4$ stages $S$. Patterns $Pj$ and $Pk$ are $(6,4,4,2,2,1,1,3,3,5,5)$ and $(3,3,5,5,7,7,6,6,4,4,2)$, correspondingly.  Arrows $\downarrow$ and $\uparrow$ show the direction of 2-position shifts of the respective patterns. $[j,k]$ denotes linear function $x_j \oplus x_{j+1} \oplus ... \oplus x_k$ of primary variables, that accepts the application of Phase gates to, so long as contained to within the circuit.  A total of $4$ Phase gate stages is required; Phase gates can be applied to the individual literals selectively in the beginning or at the end of the circuit.}
\label{fig:ccz}
\end{figure}
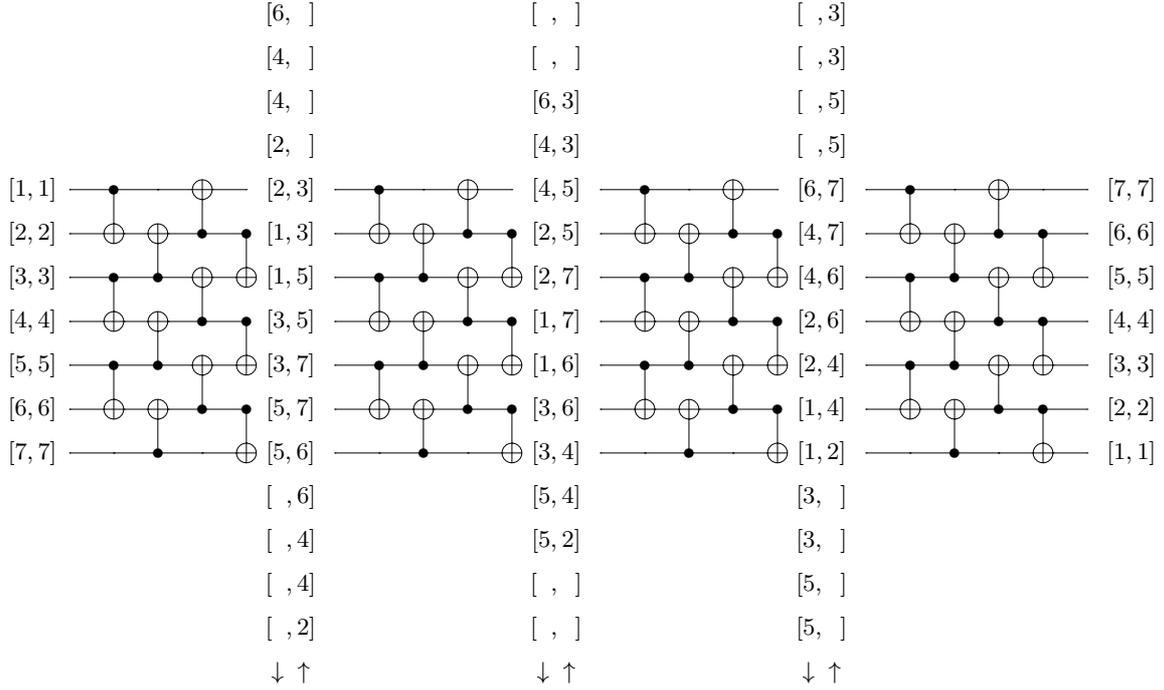

Consider phase polynomial description of the circuit $\widehat{\text{-CZ-}}$. However, rather than describe both parts of the expression, phase polynomial itself and the linear reversible transformation, over the set of primary variables, we will describe phase polynomial over variables $y_1,y_2,...,y_n$ defined as follows:
\[
y_1:=x_1, \\
y_2:=x_1 \oplus x_2, \\
..., \\
y_n:= x_1 \oplus x_2 \oplus ... \oplus x_n.
\] 
This constitutes the change of basis $\{x_1,x_2,...,x_n\} \mapsto \{y_1,y_2,...,y_n\}$.  Similarly to how it was done in the proof of Theorem \ref{thm:1}, we reduce phase polynomial representation of $\widehat{\text{-CZ-}}$ to the application of Phase gates to the EXORs of pairs and the individual variables from the set $\{y_1,y_2,...,y_n\}$, 
\begin{equation}\label{eq:alljk}
\sum\limits_{j=1}^{n}u_jy_j + \sum\limits_{j=1}^{n}\sum\limits_{k=j+1}^{n}u_{j,k}(y_j \oplus y_k),
\end{equation}
and the linear reversible function $g(x_1,x_2,...,x_n): \ket{x_1x_2...x_n} \mapsto \ket{x_nx_{n-1}...x_1}$.  Observe that $y_j \oplus y_k = x_j \oplus x_{j+1} \oplus ... \oplus x_k$, and thereby this linear function can be encoded by the integer segment $[j,k]$. The primary variable $x_j$ admits the encoding $[j,j]$.  We use this notation next.  In the following we implement the pair of the phase polynomial and the reversal of qubits (a linear reversible function) via a quantum circuit.

Observe that the swapping operation $g(x_1,x_2,...,x_n): \ket{x_1x_2...x_n} \mapsto \ket{x_nx_{n-1}...x_1}$ can be implemented as a circuit similar to the one from Theorem 5.1 \cite{ar:kms} in depth $2n+2$.  The rest of the proof concerns the ability to insert Phase gates in the circuit accomplishing the reversal of qubits such as to allow the implementation of each term in the phase polynomial $\sum\limits_{j=1}^{n}u_jy_j + \sum\limits_{j=1}^{n}\sum\limits_{k=j+1}^{n}u_{j,k}(y_j \oplus y_k)$.

Since our qubit reversal circuit is slightly different from the one used in \cite{ar:kms}, and we explore its structure more extensively, we describe it next.  It consists of $n+1$ alternating stages, $S_1$ and $S_2$, where 

\begin{tabular}{ll}
$S_1=$ & $\;\CNOT(x_1;x_2)\CNOT(x_3;x_4)...CNOT(x_{n-2};x_{n-1})$ \\ 
       & $\cdot\CNOT(x_3;x_2)\CNOT(x_5;x_4)...\CNOT(x_n;x_{n-1}) \text{ for odd } n \text{, and}$
\end{tabular}

\begin{tabular}{ll}
$S_1=$ & $\;\CNOT(x_1;x_2)\CNOT(x_3;x_4)...CNOT(x_{n-1};x_n)$ \\ 
       & $\cdot\CNOT(x_3;x_2)\CNOT(x_5;x_4)...\CNOT(x_{n-1};x_{n-2}) \text{ for even } n$ 
\end{tabular}

\noindent is a depth-$2$ circuit composed with the $\CNOT$ gates, and 

\begin{tabular}{ll}
$S_2=$ & $\;\CNOT(x_2;x_1)\CNOT(x_4;x_3)...CNOT(x_{n-1};x_{n-2})$ \\ 
& $\cdot\CNOT(x_2;x_3)\CNOT(x_4;x_5)...\CNOT(x_{n-1};x_{n})$ for odd $n$, and
\end{tabular}

\begin{tabular}{ll}
$S_2=$ & $\;\CNOT(x_2;x_1)\CNOT(x_4;x_3)...CNOT(x_{n};x_{n-1})$ \\
& $\cdot\CNOT(x_2;x_3)\CNOT(x_4;x_5)...\CNOT(x_{n-2};x_{n-1})$ for even $n$ 
\end{tabular}

\noindent is also a depth-$2$ circuit composed with the $\CNOT$ gates.  We refer to the concatenation of $S_1$ and $S_2$ as $S$.  The goal is to show that after $\lceil \frac{n}{2} \rceil$ applications of the circuit $S$ we are able to cycle through all $\frac{n(n+1)}{2}$ linear functions $[j,k]$, $j \leq k$. 

The remainder of the proof works slightly differently depending on the parity of $n$.  First, choose odd $n=2m+1$.  Consider two patterns of length $2n-3$,
\begin{eqnarray*}
Pj=(n-1,n-3,n-3,...,4,4,2,2,1,1,3,3,...,n-2,n-2) \text{ and} \\
Pk=(3,3,5,5,...,n,n,n-1,n-1,n-3,n-3,...6,6,4,4,2).
\end{eqnarray*} 
Observe by inspection that the $i^{\text{th}}$ linear function computed by the single application of the stage $S$ is given by the formula $[Pj(n-3+i),Pk(i)]$, where $Pj(l)$ and $Pk(l)$ return $l^{\text{th}}$ component of the respective pattern.  It may further be observed, via direct multiplication by the linear reversible matrix corresponding to the transformation $S$, that the $i^{\text{th}}$ component upon $t$ ($t\leq m$) applications of the circuit $S$ is computable by the following formula, $[Pj(n-1-2t+i),Pk(2t-2+i)]=[Pj(n-3-2(t-1)+i),Pk(2(t-1)+i)]$. A simple visual explanation can be given: at each application of $S$ pattern $Pj$ is shifted by two positions to the left (down, Figure \ref{fig:ccz}), whereas pattern $Pk$ gets shifted by two positions to the right (up, Figure \ref{fig:ccz}).

Observe that every $[j,k]$, $j=1..n, k=1..n, j\leq k$ is being generated. Indeed, a given $[j,k]$ may only be generated at most once by the $0$ to $m$ applications of the circuit $S$. This is because once a given $j$ meets a given $k$ for the first time, at each following step, the respective value $k$ gets shifted away from $j$ to never meet again.  We next employ the counting argument to show that all functions $[j,k]$ are generated.  Indeed, the total number of functions generated by $0$ to $m$ applications of the stage $S$ is $(m+1)n=\left(\frac{n-1}{2}+1\right)n=\frac{n(n+1)}{2}$, each linear function generated is of the type $[j,k]$ ($j=1..n, k=1..n, j\leq k$), none of which can be generated more than once, and their total number is $\frac{n(n+1)}{2}$.  This means that every $[j,k]$ is generated. 

We illustrate the construction of the circuit implementing $\widehat{\text{-CZ-}}$ for $n=7$ in Figure \ref{fig:ccz}.

For even $n=2m$ the construction works similarly. The patterns $Pj$ and $Pk$ are $(n,n-2,n-2,n-4,n-4,...,2,2,1,1,3,3...,n-3,n-3,n-1)$ and $(3,3,5,5,...,n-1,n-1,n,n,n-2,n-2,...,4,4,2,2)$, respectively. The formulas for computing linear function $[j,k]$ for $i^{\text{th}}$ coordinate after $t$ applications of $S$ is $[Pj(n-2t+i),Pk(2t-2+i)]$. After $m$ applications of the circuit $S$ we generate linear functions $x_n, x_{n-1}, ..., x_4, x_2$ in addition to the $m$ new linear functions of the form $[j,k]$ ($j<k$). 
\end{proof}

Circuit depth makes most sense when applied to measure depth across most computationally intensive operations.  In both of the two leading approaches to quantum information processing, and limiting the attention to fully programmable universal quantum machines, superconducting \cite{www:IBM} and trapped ions \cite{ar:deb}, the two-qubit gates take longer to execute and are associated with lower fidelity.  As such, they constitute the most expensive resource and motivate our choice to measure depth in terms of the two-qubit operations.  The selection of the LNN architecture to measure the depth over is motivated by the desire to restrict arbitrary interaction patterns to a reasonable set.  Both superconducting and trapped ions qubit-to-qubit connectivity patterns  \cite{www:IBM, ar:deb} are furthermore such that they allow embedding the linear chain in them.  

A further concern is that the two-qubit $\CNOT$ gate may not be native to a physical implementation, and therefore the $\CNOT$ implementation may likely use correcting single-qubit gates before and after using a specific two-qubit interaction.  This means that interleaving the two-qubit gates with the single-qubit gates such as done in the proof of Theorem \ref{thm:2} may not increase the depth, and restricting depth calculation to just the two-qubit stages is appropriate.  We did, however, report enough detail to develop depth figure over both single- and two-qubit gates for the implementations of stabilizer circuits relying on our construction. 

\begin{corollary}\label{col:2}
Arbitrary $n$-qubit stabilizer unitary can be executed in two-qubit depth $14n-4$ as a $\{\P, \H, \CNOT\}$ circuit over LNN architecture.
\end{corollary}
\begin{proof}
Firstly, observe that -H-C-CZ-P-H-P-CZ-C $=$ -H-C$\widehat{\text{-CZ-}}$P-H-P$\widehat{\text{-CZ-}}$C-. This is because both $\widehat{\text{-CZ-}}$ stages reverse the order of qubits, and therefore the effect of the qubit reversal cancels out. The two-qubit depth of the -C- stage is $5n$ \cite{ar:kms}, and the two-qubit depth of the $\widehat{\text{-CZ-}}$ stage is $2n+2$, per Theorem \ref{thm:2}.  This means that the overall two-qubit depth is $14n+4$.  This number can be reduced somewhat by the following two observations.  Name individual stages in the target decomposition as follows, -H-C$_1$$\widehat{\text{-CZ$_1$-}}$P-H-P$\widehat{\text{-CZ$_2$-}}$C$_2$-. 
Using the construction in Theorem \ref{thm:2}, we can implement $\widehat{\text{-CZ$_1$-}}$ without the first $S$ circuit through applying Phase gates at the end of it (see Figure \ref{fig:ccz} for illustration).  The first $S$ circuit can then be combined with the -C$_1$- stage preceding it.  This results in the saving of $4$ layers of two-qubit computations.  Similarly, $\widehat{\text{-CZ$_2$-}}$ can be implemented up to $S$ if it is implemented in reverse, and phases are applied in the beginning (the end, but invert the circuit).  This allows to merge depth-$4$ computation $S$ with the stage -C$_2$- that follows.  These two modifications result in the improved depth figure of $14n-4$.  
\end{proof}

Observe how the aggregate contribution to the depth from both -CZ- stages used in this paper, $\sim \!4n$, is less than that from a single -C- stage, $5n$.  The result of \cite{ar:kms} can be applied to the 11-stage decomposition -H-C-P-C-P-C-H-P-C-P-C- of \cite{ar:ag} to obtain a two-qubit depth-$25n$ LNN-executable implementation of an arbitrary stabilizer unitary.  In comparison, our reduced 8-stage decomposition -H-C-CZ-P-H-P-CZ-C- allows execution in the LNN architecture in only $14n-4$ two-qubit stages.

\section{Discussion}
In this paper, we reduced the 11-stage computation -H-C-P-C-P-H-P-C-P-C- \cite{ar:ag} to an 8-stage computation of the form -H-C-CZ-P-H-P-CZ-C-.  The optimized implementation relies on the stage -CZ- not previously considered by \cite{ar:ag}.  We showed evidence that the -CZ- stage is likely superior to the comparable -C- stage.  Indeed, the number of the Boolean degrees of freedom in the -CZ- stage is only about a half of that in the -C- stage, two-qubit gate counts for optimal implementations of -CZ- circuits remain smaller than those for -C- circuits (see Table \ref{tab:main}), and -CZ- computations were possible to implement in a factor of $2.5$ less depth than that for -C- computations over LNN architecture.  

We reported a two-qubit depth-$(14n-4)$ implementation of stabilizer unitaries over the gate library $\{\P,\H,\CNOT\}$, executable in the LNN architecture.  This improves previous result, a depth-$25n$ circuit \cite{ar:ag,ar:kms} executable over LNN architecture.

Our 8-stage construction can be written in $16$ different ways, by observing that -C-CZ-P- can be written in $4$ different ways (-C-CZ-P-, -C-P-CZ-, -P-CZ-C-, and -CZ-P-C-). 

For the purpose of practical implementation we believe a holistic approach to the implementation of the 3-layer stage -P-CZ-C- may be due, where the linear reversible function $g(x_1,x_2,...,x_n)$ is implemented by the $\CNOT$ gates such that the intermediate linear Boolean functions generated go through the set that allows implementation of the phase polynomial part. 

\section{Acknowledgements}
I thank Dr. Yunseong Nam from the University of Maryland--College Park and Dr. Martin Roetteler from Microsoft Research--Redmond for their helpful discussions.

Circuit diagrams were drawn using qcircuit.tex package, \href{http://physics.unm.edu/CQuIC/Qcircuit/}{http://physics.unm.edu/CQuIC/Qcircuit/}.

This material was based on work supported by the National Science Foundation, while working at the Foundation. Any opinion, finding, and conclusions or recommendations expressed in this material are those of the author and do not necessarily reflect the views of the National Science Foundation.

\end{document}